\documentclass[11pt]{article}
\usepackage{float, hyperref, graphicx, amsfonts, amsthm, amsmath, amssymb}
\theoremstyle{plain}
\newtheorem{theorem}{Theorem}[section]
\newtheorem{lemma}[theorem]{Lemma}
\newtheorem{claim}[theorem]{Claim}

\newtheorem{corollary}[theorem]{Corollary}
\newcommand{\arc}[1]{\overrightarrow{#1}}
\newcommand{\diam}{\operatorname{diam}}
\newcommand{\dist}{\operatorname{dist}}
\newcommand{\improve}{\operatorname{improve}}

\title{On a Bounded Budget Network Creation Game\thanks{A preliminary version of this paper appeared in Proceedings of the 23rd ACM Symposium on Parallelism in Algorithms and Architectures (SPAA 2011), pp.~207--214.}}

\author{Shayan Ehsani\thanks{Current address:
   Department of Management Science and Engineering, Stanford University, email: \texttt{shayane@stanford.edu}.} \\
   Department of Computer Engineering\\
   Sharif University of Technology \\
   \texttt{ehsani@ce.sharif.edu}\\
   \and
   Saber Shokat Fadaee\thanks{Current address: College of Computer and Information Sciences, Northeastern University, email: \texttt{saber@ccs.neu.edu}.}\\
   Department of Computer Engineering\\
   Sharif University of Technology\\
   \texttt{shokat@ce.sharif.edu}\\
   \and
   MohammadAmin Fazli\\
   Department of Computer Engineering\\
   Sharif University of Technology\\
   \texttt{fazli@ce.sharif.edu}\\
   \and
   Abbas Mehrabian\\
   Department of Combinatorics and Optimization\\
   University of Waterloo\\
   \texttt{amehrabi@uwaterloo.ca}\\
   \and
   Sina Sadeghian Sadeghabad\thanks{Current address: Department of Combinatorics and Optimization, University of Waterloo, email: \texttt{s3sadegh@uwaterloo.ca}.}\\
   Department of Computer Engineering\\
   Sharif University of Technology\\
   \texttt{s\_sadeghian@ce.sharif.edu}\\
   \and
   MohammadAli Safari\\
   Department of Computer Engineering\\
   Sharif University of Technology\\
   \texttt{msafari@sharif.edu}\\
   \and
   Morteza Saghafian\\
   Department of Mathematical Sciences\\
   Sharif University of Technology\\
   \texttt{saghafian@ce.sharif.edu}\\
   }
\date{}

\begin{document}

\maketitle

\begin{abstract}
\noindent We consider a network creation game
in which each player (vertex) has a fixed budget to establish links to other players.
In our model, each link has unit price and
each agent tries to minimize its cost,
which is either its local diameter or its total distance to other players in the (undirected) underlying graph of the created network.
Two versions of the game are studied:
in the MAX version, the cost incurred to a vertex is the maximum distance between the vertex and other vertices,
and in the SUM version, the cost incurred to a vertex is the sum of distances between the vertex and other vertices.
We prove that in both versions pure Nash equilibria exist,
but the problem of finding the best response of a vertex is NP-hard.
We take the social cost of the created network to be its diameter,
and next we study the maximum possible diameter of an equilibrium graph with $n$ vertices in various cases.
When the sum of players' budgets is $n-1$,
the equilibrium graphs are always trees,
and we prove that their maximum diameter
is $\Theta(n)$ and $\Theta(\log n)$
in MAX and SUM versions, respectively.
When each vertex has unit budget (i.e.~can establish link to just one vertex),
the diameter of any equilibrium graph in either version is $\Theta(1)$.
We give examples of equilibrium graphs in the MAX version,
such that all vertices have positive budgets and yet the diameter is $\Omega(\sqrt{\log n})$.
This interesting (and perhaps counter-intuitive) result shows that increasing the budgets
may increase the diameter of equilibrium graphs and hence deteriorate the network structure.
Then we prove that every equilibrium graph in the SUM version has diameter $2^{O({\sqrt{\log n}})}$.
Finally, we show that if the budget of each player is at least $k$,
then every equilibrium graph in the SUM version is $k$-connected or has diameter smaller than 4.
\end{abstract}



\section{Introduction}
In recent years, a lot of research has been conducted on network design problems, because of their importance in computer science and operations research~\cite{christofides, johnson, cheriyan}.
The aim in this line of research is usually to build a minimum cost network that satisfies certain properties, and the network structure is usually determined by a central authority.
However, this is in contrast to many real world situations
such as social networks, client-server systems and peer-to-peer networks,
where network structures are determined in a distributed manner by selfish agents~\cite{feldman, golle, skyrms}.
The formation of these networks can be formulated as a game,
which is usually called a \emph{network creation game}.
In network creation games, as in any other game, there are selfish players that interact with each other.
Each player has its own objective, and attempts to minimize the cost it incurs in the network,
regardless of how its actions affect other agents.
The players are placed at the {nodes} of the network graph,
and can create links to other nodes with certain restrictions,
e.g.~there could be an upper bound for the number of links a player constructs.
The utility functions of the players should be defined properly
to be consistent with their natural interests,
e.g.~minimizing the cost to communicate with other players.

In network creation games,
players interact with each other by adding and removing links between themselves.
Many variants of these games arise by defining different utility functions and
possible transitions between strategies of each player.
Some authors have studied undirected graphs while others have considered directed graphs.
For undirected graphs there is an issue of ``ownership'':
when there exists a link between two nodes, but just one of the nodes
wants to keep it, is it removed from the network?
For directed graphs, the question is whether both endpoints
of a link can use it to communicate.
In some models creating a link incurs a cost to the player,
i.e.~the number of created links appears in utility functions,
while in other models restrictions for creating links appear in the set of available strategies for players.

Although the creation of a network in such a game is a dynamic process,
the structure of the resulting network (if it does converge to a stable structure)
provides  valuable information about the effectiveness of the game rules.
Thus, existence and structure of \emph{stable networks}
has been widely studied.
Jackson and Wolinsky~\cite{jackson} defined a network to be \emph{pairwise stable} if, roughly speaking,
none of the nodes are willing to delete an incident link
and no pair of non-adjacent nodes are willing to build a link between themselves.
Fabrikant et al.~\cite{fabrikant} considered \emph{Nash equilibria} of the game
as its stable states.
A Nash equilibrium, which is a well known concept in game theory,
is a state of a game in which no player can increase her utility by changing her strategy,
assuming the strategies of other players are kept unchanged.
The main difference between these two concepts is that,
when considering pairwise stability,
we are thinking of the players cooperating with each other,
whereas we think of non-cooperating players when we consider Nash equilibria.

The efficiency of the network formed by a game is measured by different factors rather than the player utilities.
We are mostly interested in measuring a global parameter,
for instance the {diameter} (the largest distance between any pair of nodes),
and the vertex connectivity (the minimum number of nodes whose removal disconnects the network)
of the network are two of the possible candidates; these are important parameters in every network.
The global parameter, called the \emph{social cost} of the network,
quantizes how effective the network is.
To measure how the efficiency of a system degrades due to selfish behavior of its agents,
we find the social cost achieved by selfish players in a stable state of the network creation game,
and calculate its ratio to the minimum social cost among all networks.
This parameter, called the \emph{price of anarchy} of the game,
was introduced by Koutsoupias and Papadimitriou~\cite{poa} and
has been the focus of study of many works on network creation (and many other) games.

In this paper we introduce and study a new class of network creation games,
which is motivated by the work Laoutaris et al.~\cite{laoutaris}.
In our model, there is an upper bound on the number of links each player can create,
hence the name \emph{bounded budget network creation game}.


\subsection{Previous Work}

Jackson and Wolinsky~\cite{jackson} were one of the first studying the stable states of networks created by selfish players.
They introduced the notion of pairwise stability and studied the efficiency of pairwise stable networks.
Fabrikant, Luthra, Maneva, Papadimitriou and Shenker~\cite{fabrikant} suggested studying Nash equilibria
instead of pairwise stable networks.
In their model the network graph is undirected, and there is a link between two nodes if at least
one of the nodes wants to create it.
There is a cost $\alpha$ for creating a link,
and the goal of each node is to minimize the sum of its distances to other nodes minus the amount she pays for creating links.
Note that every node wants to build more links to get closer to other vertices,
but, on the other hand, the more links she creates, the more money she has to pay.
They showed that for large $\alpha$,
the equilibrium graphs have few edges and are tree-like,
whereas for small values of $\alpha$, the equilibrium graphs are dense.
This implies that the structure of the equilibria is highly affected by the parameter $\alpha$.
They conjectured that there is a universal constant $A$ such that for $\alpha > A$, any equilibrium graph is a tree.
Albers, Eilts, Even-Dar, Mansour and Roditty~\cite{albers} disproved this conjecture using geometric constructions.
The results of Fabrikant et al.~were improved in \cite{corbo} and \cite{demaine}.
Demaine, Hajiaghayi, Mahini and Zadimoghaddam~\cite{demaine} found a $2^{O{(\sqrt{\log n})}}$ upper bound
on the diameter of equilibrium graphs, where $n$ is the number of players (nodes),
and conjectured that indeed the diameter of equilibrium graphs is polylogarithmic.
They also considered another version of the game,
where the utility function of each player is its maximum distance to other nodes minus the amount she pays for creating links.
Mihal\'{a}k and Schlegel~\cite{mihalak} further studied the game,
and proved that for the original (sum of distances) version,
the price of anarchy is $O(1)$ for $\alpha > 273 n$,
and for the latter (maximum distance) version,
the price of anarchy is $2^{O{(\sqrt{\log n})}}$ for all $\alpha$,
and is $O(1)$ for $\alpha > 129$.
Brandes, Hoefer, and Nick~\cite{disconnected_equilibria} studied a variant in which
the cost for a pair of disconnected pair is a finite value, as opposed to previous variants,
where this value was infinity.

Laoutaris, Poplawski, Rajaraman, Sundaram and Teng~\cite{laoutaris} introduced another variant of network creation games,
in which a budget is dedicated to each player for buying links.
That is, each player can build a certain number of links, and there is no cost term in the utility function,
and thus no parameter $\alpha$.
This is a natural and interesting perspective for formulating peer-to-peer and overlay networks.
To eliminate the intricacies with ownership of the links, they assumed that the links are directed
and can be used by one of their endpoints.
They defined the utility function of a player as its average distance to other nodes.
If all players have the same budget $k$,
they proved that Nash equilibria always exist
and that the price of anarchy is between
$c_1 \sqrt{\frac{n \textrm{log} k}{k\textrm{log} n}}$
and
$c_2 \sqrt{\frac{n\textrm{log} k}{\textrm{log}n}}$
for suitable positive constants $c_1,c_2$.
Our model is mainly motivated by their work, but we assume that links are bidirectional
and can be used by both of their endpoints.
However, one of the endpoints of a link is its ``owner'' and she is responsible for creating it.

Alon, Demaine, Hajiaghayi and Leighton~\cite{adhl} simplified Fabrikant et al.'s model by eliminating the parameter $\alpha$
and introduced basic network creation games.
In their model, each node locally tries to minimize its maximum distance or average distance
to other nodes, by swapping one incident edge at a time.
They say an undirected graph is a \emph{swap equilibrium}
if no node can increase its utility by swapping just one of its incident edges.
By bounding the possible transitions between strategies of a node,
they get a broader set of equilibria, which includes all Nash equilibria in the previous model.
Therefore, any upper bound for the price of anarchy for swap equilibria
is an upper bound for Nash equilibria as well.
Actually, considering swap equilibria seems to be more realistic too,
since each player can determine its best response in polynomial time and thus
the computational needs of each player is decreased.
Also, the removal of the parameter $\alpha$ has made the proofs cleaner and more general,
and thus we have also used this idea in our model.
The main difference between our games and basic network creation games is the ownership of the links.
In their model any of the two endpoints of a link can remove it,
whereas in our model any link is owned by one of its endpoints,
and only that node is able to remove it.
Two versions of basic creation games are considered, the SUM version and the MAX version,
depending on whether the goal of each player is to minimize its average distance or maximum distance to other nodes.
Alon~et~al.~found an upper bound of $2^{O(\sqrt{\textrm{log} n})}$ on the diameter of SUM equilibria,
which is stronger than previous bounds on models with a parameter $\alpha$.
We prove that the same upper bound holds for our model.
However, in the MAX version we observe essential differences even in tree equilibria:
in basic network creation games, the diameter of equilibria is at most $3$,
whereas in our model, we have tree equilibria with diameter $\Theta(n)$.

There are utility functions studied in the literature
other than the average or maximum distance of a node to other nodes.
We will not pursue them here but mention a few of them for completeness.
Bei, Chen, Teng, Zhang and Zhu~\cite{bei}
investigated a model in which each node aims to maximize its \emph{betweenness}.
This notion is introduced originally in social network analysis,
and roughly speaking, measures the amount of information passing through a node
among all pairwise information exchanges.
The \emph{clustering coefficient} of a node is defined as the probability that two of its randomly selected neighbors are directly connected to each other.
Brautbar and Kearn~\cite{brautbar} considered clustering coefficient as an incentive in creation of networks such as social networks, and considered network creation games in which the utility of each node equals its clustering coefficient.

In the model we introduce here, every player has a budget which determines the number of links it is able to build.
This simplifies the utility functions, but complicates the strategy set of the players.
Nevertheless, we are able to prove many results in both SUM and MAX versions of the game,
and we observe that, just like the model of Laoutaris et al.~\cite{laoutaris},
changing the node budgets significantly changes the structure of equilibrium networks.
In contrast to \cite{laoutaris}, in our model,
once a link is established, both its endpoints can use it equally.
This is a natural model in applications where the direction of links does not matter,
e.g.~computer networks.

\subsection{Our model and notation}

Let $G$ be a directed graph on $n$ vertices, and let $V(G)$ denote the vertex set of $G$.
The \emph{underlying graph} of $G$, which is an undirected graph obtained by ignoring the arc directions in $G$, is denoted by $U(G)$.
If both arcs $\arc{uv}$ and $\arc{vu}$ are in $G$, then $uv$ is a multiple edge with multiplicity 2 in $U(G)$,
which is viewed as a cycle with 2 vertices.
In the following, whenever we refer to the distance between two vertices of $G$, we mean their distance in $U(G)$.
The distance between two vertices $u$ and $v$ is denoted by $\dist(u,v)$.
If $u$ and $v$ are in different connected components of $U(G)$,
then it is natural to define their distance as infinity.
However, we define their distance to be a  large constant $C_{\textrm{inf}}$ so that the vertices
have the incentive to decrease the number of connected components.
We choose $C_{\textrm{inf}} = n^2$ for a reason that will be discussed later.
The \emph{diameter} of $G$, written $\diam(G)$, is the maximum distance between any two vertices of $G$.
For a vertex $u$ and subset $A \subseteq V(G)$,
the distance between $u$ and $A$, written $\dist(u,A)$, is defined as
$$\dist(u,A) = \min \{\dist(u,a) : a \in A\}.$$
The \emph{local diameter} of a vertex $u$ is the maximum of its distances to other vertices.
Note that if the graph is disconnected, then the local diameter of all vertices is $n^2$.

Let $n$ be a positive integer and $b_1,b_2,\dots,b_n$ be nonnegative integers less than $n$.
A \emph{bounded budget network creation game} with parameters $b_1,b_2,\dots,b_n$, denoted by $(b_{1},b_{2},\dots,b_{n})$-BG, is the following game.
There are $n$ players and the strategy of player $i$ is a subset $S_i \subseteq \{1,2,\dots,n\}\backslash\{i\}$ with $|S_i| = b_i$.
We may build a directed graph $G$ for every strategy profile $(S_1,\dots,S_n)$ of the game,
with vertex set $\{u_1,\dots,u_n\}$ and such that $\arc{u_i u_j}$ is an arc in $G$ if $j \in S_i$.
Any such graph $G$ is called a \emph{realization} of $(b_{1},b_{2},\dots,b_{n})$-BG.
We will identify each vertex with its corresponding player.
If $\arc{u_i u_j}$ is in $G$, then we say $\arc{u_i u_j}$ is \emph{owned by} vertex $u_i$.
Note that $u_i$ owns exactly $b_i$ arcs.
We think of $b_i$ as the \emph{budget} available to vertex $u_i$,
which she can use to build links to other vertices.
If both $\arc{u_i u_j}$ and $\arc{u_j u_i}$ are in $G$, then the pair $\{u_i, u_j\}$ is called a \emph{brace}.

We consider two versions of bounded budget network creation games, which differ in the definition of the cost function.
In the \emph{SUM} version, the cost incurred to each vertex is the sum of its distances to other vertices, that is, for each $u\in V(G)$,
$$c_{SUM}(u) = \sum_{v\in V(G)}{\dist(u, v)}.$$
By choosing $C_{\textrm{inf}}$ to be at least $n^2$, we ensure that for every vertex $u$, $c_{SUM}(u)$ decreases whenever
$u$ changes its strategy so that the number of vertices in its connected component is increased.
In the \emph{MAX} version, if $U(G)$ has $\kappa$ connected components,
then the cost incurred to each vertex is defined as
$$c_{MAX}(u) = \max \{\dist(u, v) : v \in V(G)\} + (\kappa-1) n^2.$$
The term $\max \{\dist(u, v) : v \in V(G)\}$ is simply the local diameter of $u$,
and the (artificial) term $(\kappa-1) n^2$ has been added to the cost function so that when the network is disconnected,
the vertices would have the incentive to decrease the number of connected components.

We say a vertex is playing its \emph{best response} if it cannot decrease its cost by changing its strategy while the other vertices' strategies are fixed.
Notice that a vertex does not need to have a unique best response.
A strategy profile is called a \emph{(pure) Nash equilibrium} if in that profile, all players are playing their best responses.
In this case the graph $G$ is said to be a \emph{Nash equilibrium graph}, or simply an \emph{equilibrium graph} for $(b_1,b_2,\dots,b_n)$-BG.
The \emph{price of stability} of $(b_1,\dots,b_n)$-BG is defined as
$$\frac{\min \{ \diam(G^{NE}) : G^{NE} \mathrm{\ is\ an\ equilibrium\ for\ } (b_1,\dots,b_n)\mathrm{-BG} \}} {\min \{ \diam(G) : G \mathrm{\ is\ a\ realization\ of\ } (b_1,\dots,b_n)\mathrm{-BG} \}}.$$
And the \emph{price of anarchy} of $(b_1,\dots,b_n)$-BG is defined as
$$\frac{\max \{ \diam(G^{NE}) : G^{NE} \mathrm{\ is\ an\ equilibrium\ for\ } (b_1,\dots,b_n)\mathrm{-BG} \}} {\min \{ \diam(G) : G \mathrm{\ is\ a\ realization\ of\ } (b_1,\dots,b_n)\mathrm{-BG} \}}.$$
The price of anarchy measures how the efficiency of the network degrades due to selfish behavior of its agents.
In this paper, networks with smaller diameter are considered more efficient,
and the social cost of a strategy profile is the diameter of the constructed graph.
It is worth noting that, if
$$b_1+b_2+\dots+b_n \geq n-1,$$
then the denominator of both of these fractions is $O(1)$ (see Theorem~\ref{thm:ne-existence}),
and the main challenge is to evaluate the nominators,
i.e.~the diameters of equilibrium graphs.
In this case, the undirected underlying graphs of equilibria are connected (see Lemma~\ref{lem:connected}).
Instances with
$$b_1+b_2+\dots+b_n < n-1$$
are not very interesting, since the constructed networks are always disconnected
and both of the fractions are equal to 1.
In this paper, all logarithms are in base 2,
and generally we do not try to optimize the constant factors.

\subsection{Our results and organization of the paper}
We study various properties of equilibrium graphs for bounded budget network creation games.
In particular, we analyze the diameter of equilibrium graphs in various special cases,
which results in bounds for the price of anarchy in these cases.
First, in Section~\ref{sec:existence}, we prove that for every nonnegative sequence
$b_1,\dots,b_n$, the game $(b_{1},b_{2},\dots,b_{n})$-BG has a Nash equilibrium in both versions,
and that the price of stability of this game is $O(1)$.
In Section~\ref{sec:trees}, we study the price of anarchy in extreme instances
in which the sum of budgets is $n-1$.
Note that this is the smallest sum needed to have a connected network.
For these instances, we prove that the price of anarchy is $\Theta(n)$ and $\Theta(\log n)$
in MAX and SUM versions, respectively.

In Section~\ref{sec:unitbudget}, we prove that the price of anarchy in instances
in which the budget of each players is equal to 1, is $\Theta(1)$ in either version.
One may expect that further increasing the players' budgets
will result in equilibrium graphs with even smaller diameters.
In Section~\ref{sec:lowermax}, we show that, interestingly,
this is not true and there exist instances
in which all players have positive budgets and the price of anarchy is $\Omega(\sqrt{\log n})$ in the MAX version.
Such a counter-intuitive behavior had been observed previously in algorithmic game theory,
and perhaps the most famous example, known as the \emph{Braess's paradox},
is given by Braess, Nagurney, and Wakolbinger~\cite{nagurney} in network routing games.
They observed that adding extra capacity to the links of a road network,
which is used by several selfish commuters,
in some cases might reduce the overall performance.

In Section~\ref{sec:sumupper}, we give a general upper bound of
$2^{O({\sqrt{\log n}})}$ for the price of anarchy in the SUM version.
Our bounds on the price of anarchy in various classes of instances
in both versions are summarized in Table~\ref{t:results}.
In Section~\ref{sec:connectivity}, we consider the connectivity of equilibrium graphs,
and prove that if the budget of each player is at least $k$,
then every equilibrium graph in the SUM version with diameter larger than 3 is $k$-connected.
We conclude with a discussion of our results and proposing some open problems in Section~\ref{sec:conclusion}.

\begin{table}
\begin{center}
 \caption{Our bounds on the price of anarchy in various classes of instances}
\begin{tabular}{|ccc|}
\hline & \textbf{MAX} & \textbf{SUM}  \\
\hline \textbf{Trees}                &  $\Theta(n)$            & $\Theta(\log n)$  \\
\hline \textbf{All-Unit Budgets}     & $\Theta(1)$             & $\Theta(1)$ \\
\hline \textbf{All-Positive Budgets} & $\Omega(\sqrt{\log n})$ & $2^{O({{\sqrt{\log n}})}}$ \\
\hline \textbf{General}              & $\Theta(n)$             & $2^{O({{\sqrt{\log n}})}}$ \\
\hline
\end{tabular}
 \label{t:results}
\end{center}
\end{table}


\section{Existence of Nash equilibria}
\label{sec:existence}

In this section,
we prove that for every nonnegative $b_1, b_2, \dots, b_n$, Nash equilibria exist for both MAX and SUM versions of $(b_1, b_2, \dots, b_n)$-BG.
Moreover, we prove that the price of stability of this game is $O(1)$.
Before proving the main result of this section,
we show that computing a player's best response in bounded budget network creation games is an intractable problem.

\begin{theorem}
The problem of finding a player's best response in both MAX and SUM versions of bounded budget network creation games is NP-Hard.
\end{theorem}
\begin{proof}
We reduce the {\sc $k$-center} problem to the problem of finding a player's best response in the MAX version of the game.
In the {\sc $k$-center} problem, a graph and a positive integer $k$ is given and the aim is to find a subset $S$ of $k$ vertices of
so as to minimize the maximum distance from a vertex to $S$, i.e.~we want to find
$$\min_{|S| = k}\ \max_{v\in V} \dist(v, S).$$

Assume that we are given an undirected graph $H$ with $n$ vertices, and we are supposed to find
an optimal solution to the {\sc $k$-center} problem.
Consider a directed graph $G$ such that $U(G) = H$, and a game $(b_1,b_2,\dots,b_{n},b_{n+1})$-BG,
where $b_i$ is the outdegree of the $i$-th vertex in $G$, and define $b_{n+1} = k$.
Now compute a best response of the $(n+1)$-th player in the MAX version of this game,
where the strategies of other players are realized by $G$.
A best response is clearly an optimal solution for the {\sc $k$-center} problem in $H$.
The proof is complete by noting that {\sc $k$-center} is NP-hard (see~\cite{kcenter} for instance).

By using exactly the same idea,
one can reduce the {\sc $k$-median} problem~(see \cite{kmedian} for the definition)
to the problem of finding a best response in the SUM version of the game.
Since the former problem is NP-hard, the latter one is NP-hard, too.
\end{proof}

For proving the main theorem of this section,
we need a lemma that gives a sufficient condition for guaranteeing that a vertex is playing its best response.

\begin{lemma} \label{l:equibconst}
Let $u$ be a vertex in a realization of a bounded budget network creation game.
If $c_{MAX}(u) \leq 2$ and $u$ is not contained in any brace,
or $c_{MAX}(u)=1$,
then $u$ is playing its best response in both MAX and SUM versions of the game.
\end{lemma}
\begin{proof}
If $c_{MAX}(u) = 1$, then it is clear that $u$ cannot decrease its cost.
Otherwise, let $V^-$ be the set of vertices that have an arc to $u$ and $V^+$ be the set of vertices that have an arc from $u$.
Since $u$ is not an endpoint of any brace, $V^+\cap V^- = \emptyset$.
It is easy to verify that no matter how $u$ plays, it always has distance one to at most  $|V^+| + |V^-|$ vertices,
and distance at least two to the rest of the vertices.
Therefore, regardless of how $u$ plays,
its cost in the MAX version will be at least $2$, and
its cost in the SUM version will be at least $2(n - 1 - |V^-| - |V^+|) + |V^+| + |V^-|$.
Hence $u$ is already playing its best response.
\end{proof}

We are now ready to prove the main theorem of this section.

\begin{theorem}
\label{thm:ne-existence}
For every nonnegative $b_1, b_2, \dots, b_n$,
Nash equilibria exist for both MAX and SUM versions of $(b_1,\dots,b_n)$-BG.
Moreover, the price of stability of this game is $O(1)$.
\end{theorem}

\begin{proof}
Let $z$ be the number of players with zero budget, and let $\sigma = b_1 + b_2 + \dots + b_n$.
Without loss of generality, assume that the $b_1,\dots,b_n$ are in nondecreasing order, that is,
$$0=b_1 = b_2 = \dots = b_z < b_{z+1} \leq b_{z+2} \dots \leq b_{n-1} \leq b_n.$$
We consider three cases.

\noindent {\bf Case 1. $\sigma \geq n-1$ and $b_n \geq z$}\\
We provide an algorithm to build a graph $G$ all of whose vertices satisfy the conditions of Lemma~\ref{l:equibconst}.
Thus $G$ is an equilibrium graph  for both versions.
Moreover, $G$ has diameter $O(1)$, which shows that the price of stability is $O(1)$ for instances satisfying $\sigma \geq n-1$ and $b_n \geq z$.
In this case we can use a single vertex to link to zero-budget vertices and keep $G$ connected.

The graph $G$ has vertex set $\{v_1,\dots,v_n\}$ and is initially empty.
Add the arcs $\arc{v_nv_1}$, $\arc{v_n v_2}$, $\dots$,$\arc{v_n v_{b_n}}$ and then the arcs
 $\arc{v_{b_n+1} v_n}$, $\arc{v_{b_n+2} v_n}$, $\dots$ ,$\arc{v_{n-1} v_n}$ to $G$.
Note that $G$ has diameter 2 at this point, but there might be vertices whose outdegrees are less than their budgets.
If $u$ is such a vertex, then add arcs from $u$ to arbitrary vertices until its outdegree equals its budget.
This operation clearly does not increase the diameter, but may create braces.
For every brace $\{u,v\}$ such that $u$ has local diameter two and there exists a vertex $w$ not adjacent to $u$,
replace the arc $\arc{u v}$ with $\arc{uw}$.
This can be done only a finite number of times, since after every replacement the number of braces decreases.
It is easy to see that the vertices of the obtained graph have the properties of Lemma \ref{l:equibconst}
and thus this graph is an equilibrium graph.

\noindent{\bf Case 2. $\sigma \geq n-1$ and $b_n < z$}\\
As in Case 1 (but using a more complicated construction), we build a graph which is an equilibrium graph
in both versions, and has diameter $O(1)$.
In this case we cannot use a single vertex to link to zero-budget vertices and keep the graph connected,
hence we should use several vertices to do this.
We would like to use as few vertices as possible, so we will focus on vertices with large degrees.
Let $t$ be the largest index with
\begin{equation*}
\label{eq:t}
b_n + b_{n-1} + \dots + b_t \geq z + n-t.
\end{equation*}
First, note that such a $t$ exists and is larger than $z$, since
$$b_n + b_{n-1}+  \dots + b_{z+1} = \sigma \geq n-1  = z + n - (z+1).$$
Second, note that $t<n$ since $b_n < z = z + n - n$.
Define
$$A = \{v_1, v_2, \dots, v_z\},\ B = \{v_{z+1}, v_{z+2}, \dots, v_t\},\ \mathrm{and}\ C = \{v_{t+1}, v_{t+2}, \dots, v_{n-1}\}.$$
Note that $A$ is the set of zero-budget vertices,
and $\{v_t\} \cup C \cup \{v_n\}$ is the set of vertices that will
connect the set $A$ to the rest of the graph.

We start with an empty graph with vertex set $A\cup B\cup C \cup \{v_n\}$,
and add arcs to it as described in the four following phases, until the outdegree of each vertex becomes equal to its budget.
A concrete example is illustrated in Figure~\ref{f:equilibriaexist}, in which
$n=22$, $z=16$, and $t=19$.

\begin{figure}
\begin{center}
\includegraphics[width=10cm]{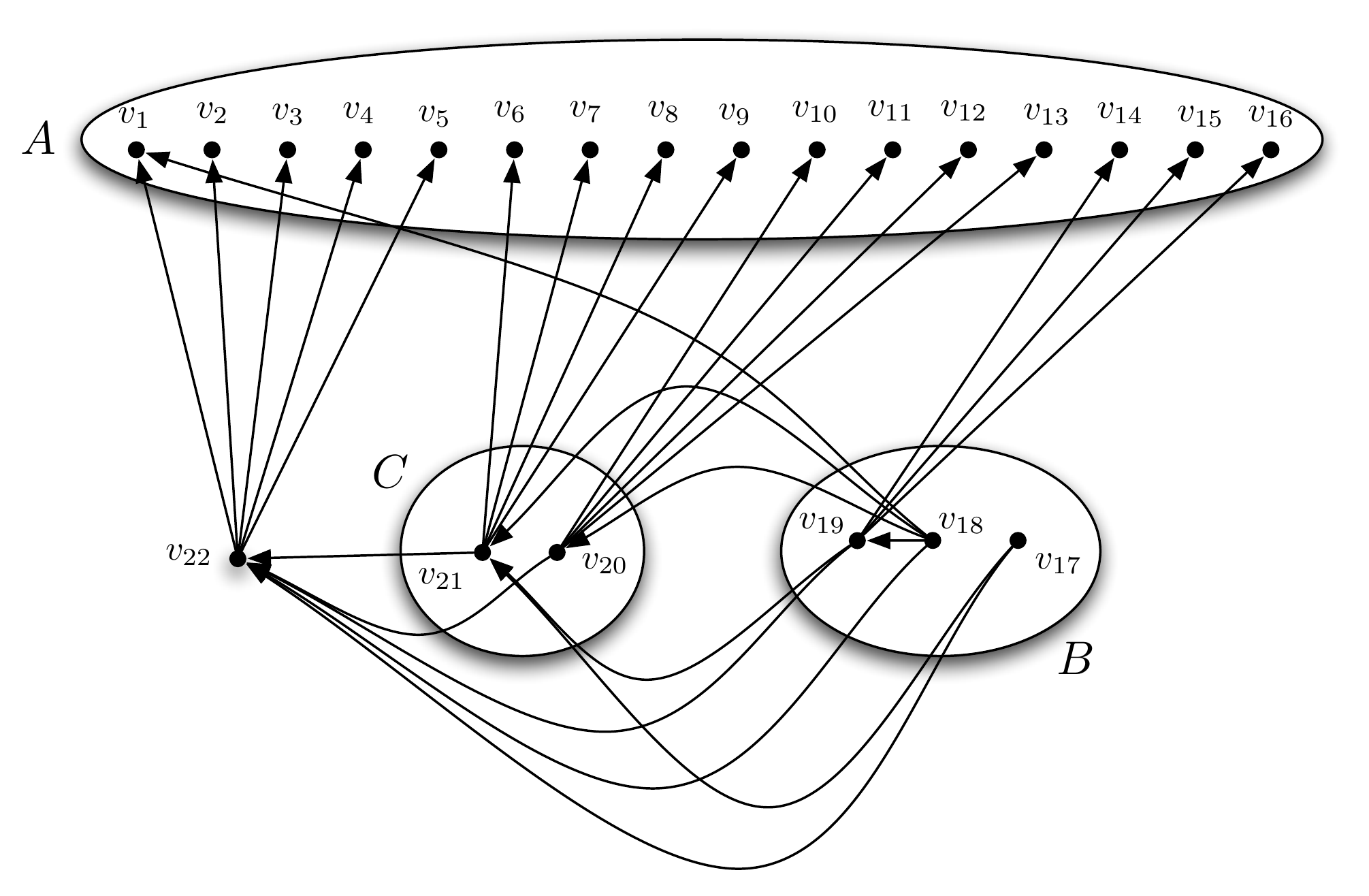}
\end{center}
\caption{Case 2 in the proof of Theorem~\ref{thm:ne-existence} }
\label{f:equilibriaexist}
\end{figure}

\begin{enumerate}
\item Add an arc from every vertex in $B\cup C$ to $v_n$
    (the arcs $\arc{v_{17} v_{22}},\arc{v_{18}v_{22}},\dots,\arc{v_{21}v_{22}}$ in Figure~\ref{f:equilibriaexist}).

\item Add arcs from $\{v_n\}\cup C \cup \{v_t\}$ to $A$:
    \begin{itemize}
    \item First, add $b_n$ arcs from $v_n$ to the first $b_n$ vertices of $A$
    (the arcs $\arc{v_{22}v_{1}},\arc{v_{22}v_2},\dots,\arc{v_{22}v_5}$ in Figure~\ref{f:equilibriaexist});
    \item Second, add $b_{n-1}-1$ arcs from $v_{n-1}$ to the next $b_{n-1}-1$ vertices of $A$
    (the arcs $\arc{v_{21}v_{6}},\dots,\arc{v_{21}v_9}$ in Figure~\ref{f:equilibriaexist});
    \item Third, add $b_{n-2}-1$ arcs from $v_{n-2}$ to the next $b_{n-2}-1$ vertices of $A$;
    (the arcs $\arc{v_{20}v_{10}},\dots,\arc{v_{20}v_{13}}$ in Figure~\ref{f:equilibriaexist});
    \item Continue similarly until you add $b_{t+1}-1$ arcs from $v_{t+1}$ to the next $b_{t+1}-1$ vertices of $A$;
    \item At last, add $s$ arcs from $v_t$ to the last $s$ vertices of $A$, where
    $$ s = z+n - (t + 1) - (b_n + \dots + b_{t+1})$$
    is positive by the definition of $t$.
    (the arcs $\arc{v_{19}v_{14}},\arc{v_{19}v_{15}},\arc{v_{19}v_{16}}$ in Figure~\ref{f:equilibriaexist}).
    \end{itemize}
    When this phase is completed, every vertex in $A$ has exactly one incoming arc.
    Moreover, the graph is a tree at this stage, and since $v_n$ has local diameter 2,
    its diameter is at most 4.

\item Add arcs from $B$ to $C\cup\{v_t\}$:
    for every vertex $u$ in $B$ whose outdegree is less than its budget, add arcs from $u$ to vertices in $C\cup\{v_t\}$
    in reverse order, i.e.~add the arcs $\arc{uv_{n-1}}, \arc{uv_{n-2}}$ and so on, until either arcs to all vertices
    of $C\cup\{v_t\}$ have been added, or the outdegree of $u$ equals its budget
    (the arcs $\arc{v_{17} v_{21}},\arc{v_{18}v_{21}},\arc{v_{18}v_{20}},\arc{v_{18}v_{19}},\arc{v_{19}v_{21}}$ in Figure~\ref{f:equilibriaexist}).

\item Add arcs from $B$ to $A$:
    for every vertex $u$ in $B$ whose outdegree is still less than its budget, add arcs from $u$ to vertices in $A$
    in order, i.e.~add the arcs $\arc{uv_1}, \arc{uv_2}$ and so on, until the outdegree of $u$ equals its budget
    (the arc $\arc{v_{18}v_1}$ in Figure~\ref{f:equilibriaexist}).

\end{enumerate}

When phase 4 is completed, for every $u \in B$, since the budget of $u$ is not more than the budget of $v_n$,
the set of neighbors of $u$ in $A$ is a subset of the set of neighbors of $v_n$ in $A$.
Therefore, every vertex in $A$ that is not adjacent to
$v_n$ has only one neighbor, which is in $C\cup\{v_t\}$.
Let $\arc{wx}$ be an arc from $C$ to $A$.
This arc could have been added in phase 2 only, so $x$ is not a neighbor of $v_n$.
Thus we have the following.

\begin{claim}
\label{clm:only}
For every arc $\arc{wx}$ from $C$ to $A$, $w$ is the only neighbor of $x$.
\end{claim}


Now, we prove that every vertex is playing its best response,
concluding that the obtained graph $G$ is an equilibrium graph.
This also implies that the price of stability of this game is $O(1)$ in this case,
since $G$ has diameter at most 4 when phase 2 finishes.
Observe that we create no brace in our construction.
Vertices in $A$ are obviously playing their best strategies as their budgets are zero.
Since $v_n$ has local diameter two, it is playing its best response by Lemma~\ref{l:equibconst}.

Let $u$ be a vertex in $C$ and we need to show that $u$ is playing its best response.
Every outgoing arc from $u$ is either going to $v_n$ or to some vertex in $A$.
The latter cannot be removed by Claim~\ref{clm:only}.
It is also easy to verify that $u$ cannot decrease its cost by removing the arc $\arc{uv_n}$
and adding an arc to another vertex.

Let $u$ be a vertex in $B$.
If in phase 4 some outgoing arcs from $u$ have been added,
then in phase 3, $u$ has already been joined to all vertices in $C\cup\{v_t\}$ and so has local diameter two.
Thus in this case, vertex $u$ satisfies the conditions of Lemma~\ref{l:equibconst} and is playing its best response.
Otherwise, since $u$ is adjacent to $v_n$, it has local diameter three.
Assume that in the beginning of phase 3, the budget of $u$ minus its outdegree was $p$.
Note that $p < |C|+1$.
First, it is easy to see that $u$ has no incentive to replace its arc $\arc{u v_n}$ with any other arc.
For any $w\in C$, at least one arc was added from $w$ to $A$ in phase 2;
so by Claim~\ref{clm:only}, there is a vertex $x\in A$ such that $w$ is its only neighbor.
Therefore, $u$ cannot make its local diameter less than 3 and so it is playing its best response in the MAX version.
Also, in the SUM version, it is easy to verify that the best strategy for $u$ is
to be adjacent to the vertices with largest degrees, i.e.~$v_{n-1}, \dots, v_{n-p}$.

\noindent {\bf Case 3. $\sigma < n-1$}\\
Let $m$ be the smallest positive integer that satisfies
$$b_m + b_{m+1} + \dots + b_{n} \geq n-m.$$
Clearly $1 <  m \leq n$ and $b_1 = b_2 = \dots = b_{m-1} = 0.$
Let $G$ be a graph with vertex set $\{v_1,v_2,\dots,v_n\}$
such that the subgraph induced by $\{v_{m},v_{m+1},\dots,v_n\}$
is an equilibrium graph for $(b_{m},b_{m+1},\dots,b_n)$-BG in the SUM version and there is no other edge in $G$.
Then it is easy to verify that $G$ is an equilibrium graph for $(b_1,b_2,\dots,b_n)$-BG in both versions.
Moreover, in this case any realization of $(b_1,\dots,b_n)$-BG is disconnected and has diameter $n^2$,
which shows that the price of stability is 1.
\end{proof}


\section{The diameter of equilibrium trees}
\label{sec:trees}
If the sum of players' budgets is less than $n-1$,
then any realization of the bounded budget network creation game is disconnected and has diameter $n^2$.
So, the smallest interesting instances of the game
are those in which the players' budgets add up to $n-1$.

\begin{lemma}
\label{lem:connected}
For any nonnegative $b_1, \ldots, b_n$ for which $\sum_{i=1}^n b_i \geq n-1$,
the underlying graphs of Nash equilibria of $(b_1,\ldots, b_n)$-BG are connected.
\end{lemma}
\begin{proof}
Let $G$ be an equilibrium graph for $(b_1,\ldots, b_n)$-BG, where $\sum_{i=1}^n b_i \geq n-1$.
If $G$ is not connected, then it has a cycle $C$, where a brace is also considered a cycle.
Pick a vertex $v$ from $C$ that owns at least one arc $\arc{vw}$ of $C$,
and let $u$ be a vertex that is in a different component.
If $v$ replaces $\arc{vw}$ with $\arc{vu}$, then the number of vertices in its connected component increases,
and the total number of connected components decreases.
Thus by definition of cost functions (and since $C_{\textrm{inf}} = n^2$), the cost of $v$ decreases in either version.
So $v$ is not playing its best response in $G$, i.e.~$G$ is not an equilibrium graph, which is a contradiction.
\end{proof}


When $\sum_{i=1}^n b_i = n-1$, it can be easily seen that every equilibrium graph is a tree.
We write {\emph {Tree-BG}} for the set of instances of bounded budget network creation games
in which the sum of budgets equals $n-1$.

In this section, we study the price of anarchy of games in Tree-BG.
We prove that in the MAX version, there exist equilibrium graphs with diameter $\Theta(n)$,
so the price of anarchy is $\Theta(n)$.
In the SUM version, we prove that equilibrium graphs have diameter $O(\log n)$,
and this bound is asymptotically tight, so the price of anarchy is $\Theta(\log n)$.

\begin{theorem}
\label{thm:treemax}
In the MAX version, for infinitely many $n$,
there are Tree-BG instances that have equilibrium graphs with diameter $\Omega(n)$.
\end{theorem}

\begin{proof}
Let $k$ be a positive integer, and let $n = 3k+1$.
Define
$$X = \{x_1, x_2, \dots, x_k\},\ Y=\{y_1, y_2, \dots, y_k\},\mathrm{\ and\ } Z = \{z_1, z_2, \dots, z_k\}.$$
Let $G$ be the tree with vertex set $X \cup Y \cup Z \cup \{w\}$ and arc set
$$\{\arc{x_1x_2},\dots,\arc{x_{k-1}x_k}, \arc{y_1 y_2},\dots, \arc{y_{k-1}y_k}, \arc{z_1z_2},\dots,\arc{z_{k-1}z_k}, \arc{x_1w}, \arc{y_1 w}, \arc{z_1w} \}.$$
See Figure~\ref{fig:treemax}.
Then $G$ is a realization of a Tree-BG instance and has diameter $2k = \Theta(n)$.
To complete the proof, we need to show that $G$ is an equilibrium graph.

\begin{figure}
\begin{center}
\includegraphics[width=8cm]{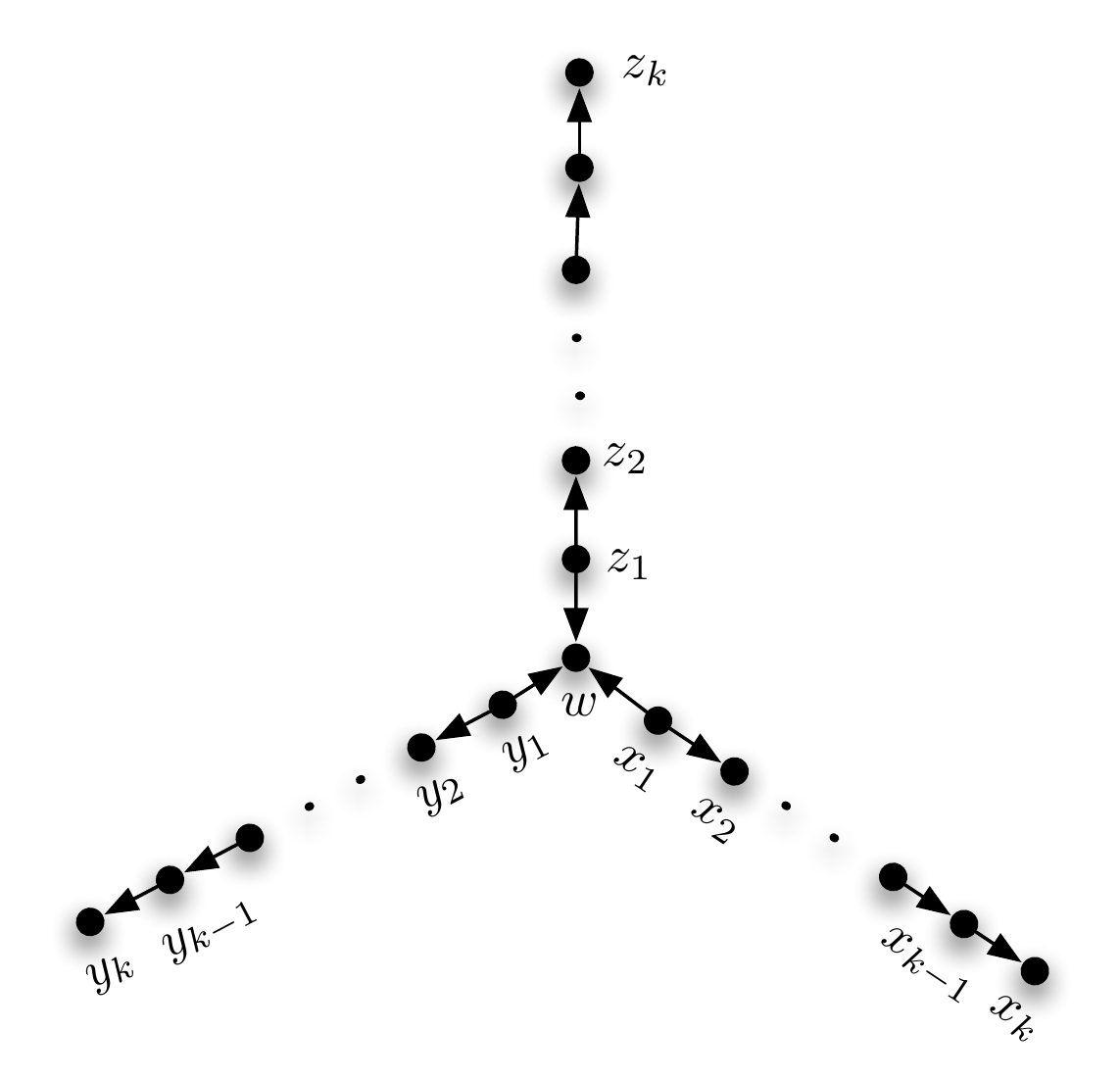}
\end{center}
\caption{Illustration of the proof of Theorem~\ref{thm:treemax} }
\label{fig:treemax}
\end{figure}

Since the vertex $w$ has no budget, it is playing its best response.
By symmetry, it is sufficient to prove that for all $1\leq i\leq k$, $x_i$ is playing its best response.
If $i > 1$, then $x_i$ has unit budget and currently has an arc to $x_{i+1}$.
If it replaces its outgoing arc with $\arc{x_i x_j}$ for some $j > i+1$, then its cost does not change at all.
If it replaces its outgoing arc with any other arc, then the graph gets disconnected,
and the cost of $x_i$ becomes $2n^2$.
Thus $x_i$ is playing its best response.

Now, let $i=1$. Note that $x_1$ has budget 2.
Clearly in order to keep the graph connected, $x_1$ should have arcs to one vertex from each of the two disjoint paths $x_2x_3\dots x_k$ and $z_kz_{k-1}\dots z_1wy_1y_2\dots y_k$.
Therefore, to minimize its local diameter,
its best response is to choose the middle of the second path (which is $w$) and an arbitrary vertex in the first path.
Thus $x_1$ is playing its best response, and the proof is complete.
\end{proof}

Next we show that the diameters of equilibrium graphs in the SUM version are much smaller.

\begin{theorem} \label{t:treepoa}
Any equilibrium graph for a Tree-BG instance in the SUM version has diameter $O(\log n)$.
\end{theorem}
\begin{proof}
Let $G$ be an equilibrium graph for a Tree-BG instance in the SUM version.
Let $d$ be the diameter of $G$ and $P = v_0v_1\dots v_d$ be a longest path in $G$.
At least half of the arcs in $P$ are in the same direction along $P$.
By symmetry, we may assume that these are the arcs $\arc{v_{i_1}v_{i_1+1}},\arc{v_{i_2}v_{i_2+1}},\dots,\arc{v_{i_{t}}v_{i_{t}+1}}$,
where $t\geq d/2$.
Every vertex not in $P$ is connected to $P$ via a unique path.
Let $A_i$ be the set of vertices that are connected to $P$ through $v_i$ (including $v_i$ itself), and let $a(i) = |A_i|$.
See Figure~\ref{f:treepoa} for an example.
Notice that all $a(i)$'s are positive since $v_i \in A_i$, and that all vertices appear in exactly one of the sets $A_{i}$.

\begin{figure}
\begin{center}
\includegraphics[width=8cm]{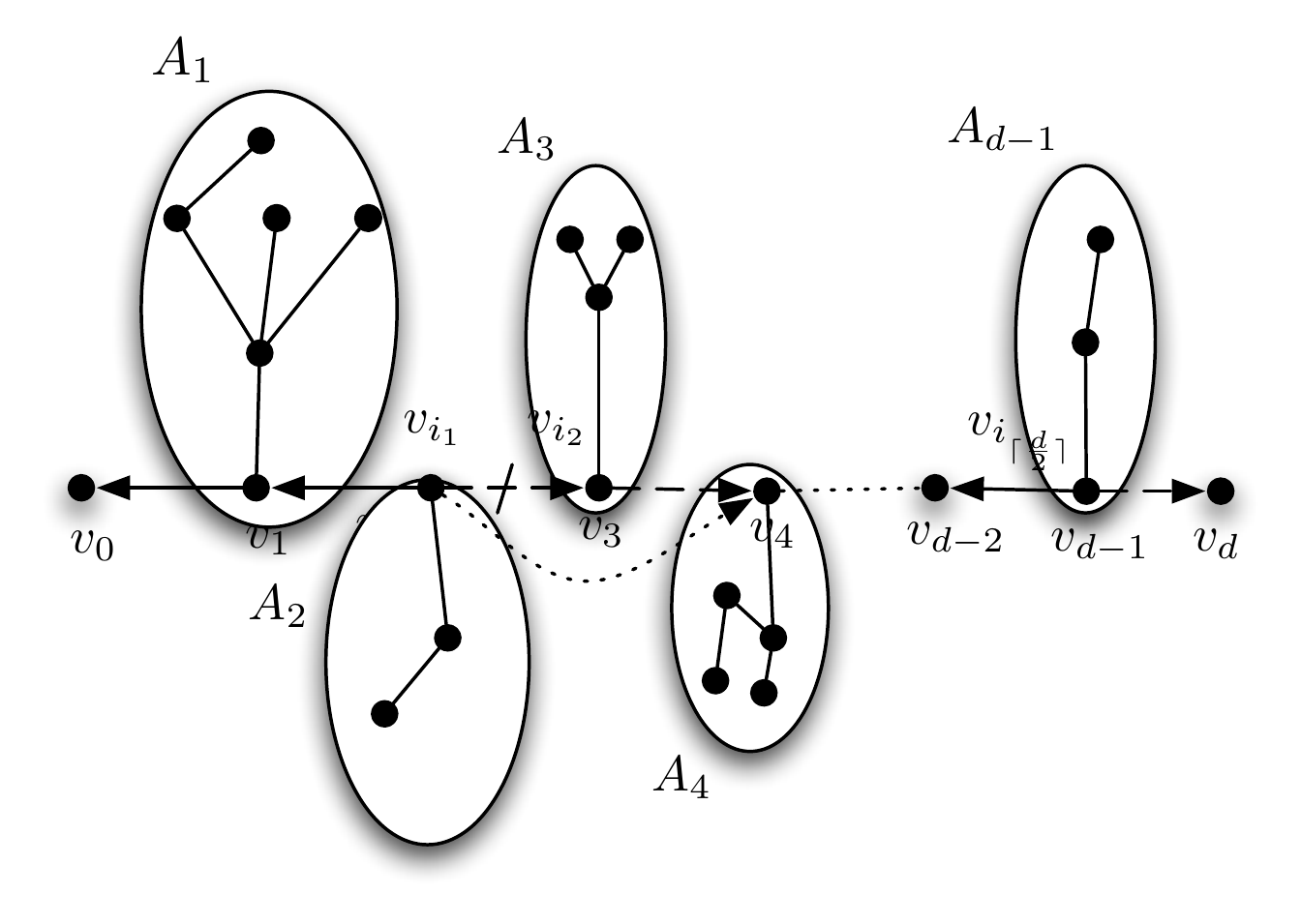}
\end{center}
\caption{Illustration of the proof of Theorem~\ref{t:treepoa}}
\label{f:treepoa}
\end{figure}

For $1\leq j< t$, if $v_{i_j}$ replaces its arc $\arc{v_{i_j} v_{i_j+1}}$ with the arc $\arc{v_{i_j} v_{i_j+2}}$,
then its distances to vertices in $A_{i_j+1}$ increase by one,
and its distances to vertices in $A_k$, $k > i_j+1$, decrease by one,
and its distances to other vertices do not change.
Since $v_{i_j}$ is playing its best response, we have
\begin{equation}
\label{eq:aij}
a(i_{j}+1) \geq \sum_{k=i_{j}+2}^{d}{a({k})} \geq \sum_{l=j+1}^{t}{a({i_{l}+1})} \quad \forall\, 1\leq j< t.
\end{equation}
Setting $j=t-1$ in (\ref{eq:aij}) gives
$$a(i_{t-1} + 1) \geq a(i_{t} + 1) \geq 1.$$
Setting $j=t-2$ in (\ref{eq:aij}) gives
$$a(i_{t-2} + 1) \geq a(i_{t-1} + 1) + a(i_{t} + 1) \geq 1 + 1 \geq 2.$$
Setting $j=t-3$ in (\ref{eq:aij}) gives
$$a(i_{t-3} + 1) \geq a(i_{t-2} + 1) + a(i_{t-1} + 1) + a(i_{t} + 1) \geq 2 + 1 + 1 \geq 4.$$
Continuing similarly, we find that $a(i_j+1) \geq 2^{t - j - 1}$ for $1\leq j < t$.
Therefore,
$$ \sum_{j=1}^{t-1} a(i_j+1) \geq \sum_{j=1}^{t-1}2^{t - j - 1} = 2 ^{t-1} - 1.$$

On the other hand, since all vertices appear in exactly one of the sets $A_{i}$, we have
$$n \geq \sum_{i=1}^{d} a(i) \geq \sum_{j=1}^{t-1} a(i_j+1) \geq 2^{t-1}-1.$$
Therefore $d \leq 2t = O(\log n)$.
\end{proof}

The bound $O(\log n)$ proved in the above theorem is tight up to constant factors,
as we next prove that there exist Tree-BG instances having equilibrium graphs with diameter $\Theta(\log n)$.

\begin{theorem} \label{t:lowerboundtree}
In the SUM version, for infinitely many $n$, there exist instances of Tree-BG that have
an equilibrium graph with diameter $\Theta(\log(n))$.
\end{theorem}
\begin{proof}
Let $k$ be a positive integer, and let $n=2^{k+1}-1$.
Let $G$ be a perfect binary tree on $n$ vertices;
that is, $G$ has vertex set $\{u_1,\dots,u_n\}$,
and for all $1\leq i < n/2$, vertex $u_i$ has arcs to vertices $u_{2i}$ and $u_{2i+1}$.
Then $G$ is a realization of a Tree-BG instance and has diameter $\Theta(\log n)$.
To complete the proof we just need to show that $G$ is an equilibrium graph in the SUM version.

For each $i$, let $T_i$ be the tree rooted at vertex $i$.
For each $1\leq i < n/2$, vertex $u_i$ has budget 2.
In order for the graph to be connected, $u_i$
must have an arc to a vertex in $T_{2i}$ and an arc to a vertex in $T_{2i+1}$.
Observe that for every $1\leq j \leq n$, vertex $u_j$ has less total distance to vertices in $T_j$ than any other vertex in $T_j$.
So, the best response for $u_i$ is to have arcs to vertices $u_{2i}$ and $u_{2i+1}$,
and thus it is already playing its best response.
For $i>n/2$, vertex $u_i$ has zero budget, so it is obviously playing its best response.
Therefore, all vertices are playing their best responses, and $G$ is an equilibrium graph.
\end{proof}


\section{The structure of equilibrium graphs for $(1,1,\dots,1)$-BG}
\label{sec:unitbudget}
In the previous section we considered bounded budget network creation games
in which the sum of players' budgets is $n-1$.
In these instances, there is at least one player with zero budget.
One might expect that if all players have positive budgets,
then the diameter of equilibrium graphs drop significantly.
In this section we consider an extreme case,
for which this expectation is realistic,
and in fact the diameter is $O(1)$.
More precisely, we study games
in which all vertices have unit budgets, i.e.~$b_i=1$ for all $i$,
and prove that the equilibrium graphs for these instances have a special structure.
In particular, we prove that all equilibrium graphs for $(1,\dots,1)$-BG in SUM and MAX versions have diameter less than 5 and 8, respectively, and therefore these instances have bounded price of anarchy in both versions.

\begin{theorem} \label{thm:unitbudgetsum}
Any equilibrium graph for $(1,\dots,1)$-BG in the SUM version is connected, has a unique cycle with at most 5 vertices,
and any vertex is either on the cycle or has a neighbor in the cycle.
\end{theorem}
\begin{proof}
Let $G$ be an equilibrium graph for $(1,\dots,1)$-BG in the SUM version.
If the number of players is two, then the only realization of the game consists of a 2-cycle and the proof is complete.
So we may assume that $n$ is larger than two.
First, we show that $G$ does not have a brace.
Assume that there is a brace $\{u,v\}$.
As $n>2$, there is a third vertex $w$ such that at least one of $u$ or $v$, say $u$, is not adjacent to $w$ in $U(G)$.
So if $u$ replaces its arc $\arc{uv}$ with the arc $\arc{uw}$, its cost decreases, which is a contradiction.
Hence $G$ does not have a brace.

Second, $U(G)$ is connected by Lemma~\ref{lem:connected}, and has $n$ edges, thus it has exactly one cycle.
Indeed, since every vertex in $G$ has outdegree 1, $G$ has a unique directed cycle.
Let $C = (v_1,v_2,\dots, v_k)$ be the unique directed cycle in $G$.
For the ease of notation, define $v_0 = v_k$.
Every vertex not in $C$ is connected to $C$ via a unique path.
Let $A_i$ be the set of vertices that are connected to $C$ through $v_i$ (including $v_i$ itself).
We may assume, by relabeling the vertices of $C$ if necessary, that $|A_k| \geq |A_{k-1}|$.

Third, we show that $k$ is at most 5. Suppose that $k>5$.
If $v_{k-2}$ replaces its arc $\arc{v_{k-2} v_{k-1}}$ with the arc $\arc{v_{k-2} v_k}$,
then its distances to vertices in $A_k$ and to $v_1$ decrease by one (as $k>5$);
and the only increment in the cost of $v_{k-2}$ would be because its distances to vertices in $A_{k-1}$ are increased by one.
Recall that $|A_{k-1}| < |A_k| + 1$, so this swap would decrease the cost of $v_{k-2}$,
which contradicts the assumption that $G$ is an equilibrium graph. Hence $k \leq 5$.

Fourth, to complete the proof, we show that for all $1\leq i \leq k$, every vertex in $A_i$ is either equal or adjacent to $v_i$.
Assume that this is not the case for some $i$, and let $v$ be a vertex in $A_i$ with maximum distance from $v_i$.
Let $l=\dist(v,v_i)$. So by the assumption, $l>1$.
Note that the subgraph of $G$ induced by $A_i$ is a tree $T$ in which  all arcs are directed toward $v_i$.
Make $T$ rooted by setting $v_i$ as the root.
Let $p$ be the parent of $v$, and $p'$ be the parent of $p$ in $T$.
Since $v$ is not adjacent to $v_i$, we have $p\neq v_i$ and $p'$ is well defined.
Let $W$ be the set of children of $p$ in $T$, and let $F=W\cup\{p\}\setminus\{v\}$.
If $v$ replaces its arc $\arc{vp}$ with the arc $\arc{vp'}$,
then its distances to vertices in $F$ increase by one,
and its distances to all other vertices decrease by one,
since vertices of $W$ have no children.
As $v$ is playing its best response, we have
\begin{equation}
\label{eq:fnf-1}
|F| \geq n - |F| - 1.
\end{equation}
If $v_{i-1}$ replaces its arc $\arc{v_{i-1}v_i}$ with the arc $\arc{v_{i-1}p}$, its distances to vertices in $F\cup\{v\}$
decrease by $l-1$, and its distance to any other vertex increases by at most $l-1$.
Since $v_{i-1}$ is playing its best response, we have
$$(n - |F| - 2)(l-1) \geq (|F| + 1)(l-1),$$
which contradicts (\ref{eq:fnf-1}) as $l>1$.
\end{proof}

Next we prove a similar structure theorem for equilibria of $(1,\dots,1)$-BG in the MAX version.

\begin{theorem} \label{thm:unitbudgetmax}
Any equilibrium graph for $(1,\dots,1)$-BG in the MAX version is connected,
has a unique cycle with at most 7 vertices,
and all vertices are within distance 2 of the cycle.
\end{theorem}
\begin{proof}
Let $G$ be an equilibrium graph for $(1,\dots,1)$-BG in the MAX version.
By Lemma~\ref{lem:connected}, $U(G)$ is connected, so it has exactly one cycle.
In fact $G$ has a unique directed cycle.
($G$ may have a brace, which is simply a cycle with two vertices.)

Let $C = (v_1,v_2,\dots, v_k)$ be the unique directed cycle in $G$.
We show that $k$ is at most 7. Suppose that $k>7$.
Every vertex not in $C$ is connected to $C$ via a unique path.
Let $A_i$ be the set of vertices that are connected to $C$ through $v_i$ (including $v_i$ itself).
Define
$$m_i = \max \{\dist(v,v_i) : v\in A_i\}.$$
We may assume, by relabeling the vertices of $C$ if necessary, that $m_{1 + \lfloor{k/2}\rfloor}$ is the largest $m_i$.
Then the local diameter of $v_1$ is exactly $\lfloor{k/2}\rfloor + m_{1 + \lfloor{k/2}\rfloor}$.
It can be verified that if $v_1$ replaces its arc $\arc{v_1 v_2}$ with the arc $\arc{v_1 v_4}$,
then as $k>7$, the distance between $v_1$ and any vertex $u \in A_j$ becomes at most $\lfloor{k/2}\rfloor -1+ \dist(v_j,u)$.
We have
$$\left\lfloor\frac{k}{2}\right\rfloor -1+ \dist(v_j,u) \leq \left \lfloor \frac{k}{2} \right \rfloor - 1 + m_j
< \left \lfloor \frac{k}{2} \right \rfloor + m_{1 + \lfloor{k/2}\rfloor},$$
which contradicts the assumption that $v_1$ is playing its best response in $G$. Hence $k \leq 7$.

Finally, to complete the proof, we show that for all $1\leq i \leq k$, every vertex in $A_i$ is within distance at most 2 from $v_i$.
Assume that this is not the case for some $i$, and let $v$ be a vertex in $A_i$ with maximum distance from $v_i$.
Note that the subgraph of $G$ induced by $A_i$ is a tree $T$ in which  all arcs are directed toward $v_i$.
Make $T$ rooted by setting $v_i$ as the root.
Let $p$ be the parent of $v$, and $p'$ be the parent of $p$ in $T$.
Since $v$ is not adjacent to $v_i$, we have $p\neq v_i$ and $p'$ is well defined.
Since $\dist(v,v_i) \geq 3$, the local diameter of $v$ is larger than 3.
If $v$ replaces its arc $\arc{vp}$ with the arc $\arc{vp'}$, then its distances to all vertices except the children of $p$ decreases,
and its distance to each children of $p$ (other than $v$ itself) increases by 1 and becomes at most 3.
Consequently, the local diameter of $v_1$ decreases, contradicting the fact that $v$ is playing its best response in $G$.
\end{proof}


\section{A lower bound for the price of anarchy in the MAX version}
\label{sec:lowermax}
We saw in the previous section that if all players have budget 1,
then the equilibrium graphs have diameter $O(1)$.
It appears intuitive that increasing the budgets would decrease the diameter of the equilibrium graphs.
However, this is not true, and in this short section we prove that for some positive budget values,
there exist equilibrium graphs in the MAX version with diameter $\Omega(\sqrt{\log n})$.
This surprising phenomenon resembles Braess's paradox in network routing games~\cite{nagurney}.
This result implies that the price of anarchy of bounded budget network creation games
when all players have positive budgets is $\Omega(\sqrt{\log n})$ in the MAX version.

\begin{lemma}
\label{lem:ADelta}
Let $U$ be an undirected graph with $n$ vertices, diameter $d$ and maximum degree $\Delta$ satisfying
$\Delta^d - 1 < n (\Delta-1)$.
Then for any vertex $v$ and any subset $A$ of vertices having size at most $\Delta$,
there exists a vertex $u$, different from $v$, with $\dist(u,A) > d-2$.
\end{lemma}

\begin{proof}
There are at most $|A|\Delta$ vertices whose distance from $A$ is exactly 1.
Similarly, there are at most $|A|\Delta^2$ vertices with distance exactly 2 from $A$.
Continuing in the same way, we find that there are at most $|A|\Delta^{d-2}$ vertices with distance exactly $d-2$ from $A$.
If there is no $u\neq v$ with $\dist(u,A) > d-2$, then we must have
$$n \leq 1 + |A| + |A|\Delta + \dots + |A|\Delta^{d-2} \leq 1 + \Delta + \Delta^2 + \dots + \Delta^{d-1} = \frac{\Delta^{d}-1}{\Delta-1},$$
which contradicts the assumption $\Delta^d - 1 < n (\Delta-1)$.
Thus there exists a vertex $u$, different from $v$, with $\dist(u,A) > d-2$.
\end{proof}

\begin{lemma}
\label{lem:wordgraph}
For every integers $t,k>3$ satisfying $(2t)^k - 1 < t^k (2t-1)$,
there exists an undirected graph $U$ with $t^k$ vertices, minimum degree at least 2, and diameter $k$,
such that every $G$ with $U=U(G)$ is an equilibrium graph in the MAX version.
\end{lemma}

\begin{proof}
Let $U$ be the graph with vertex set $\{1,2,\dots,t\}^k$ and with
vertices $(x_1,x_2,\dots,x_k)$ and $(y_1,y_2,\dots,y_k)$ being adjacent if at least one of the following happens.
\begin{enumerate}
\item $x_i=y_{i+1}$ for all $1\leq i\leq k-1$,
\item $y_i=x_{i+1}$ for all $1\leq i \leq k-1$.
\end{enumerate}
Note that we want $U$ to be a simple graph,
so we only add edges between distinct vertices,
and add at most one edge between any pair.
Then $U$ has minimum degree at least $t-1$, maximum degree $2t$, and $t^k$ vertices.
The local diameter of every vertex is $k$: for an arbitrary $(x_1,\dots,x_k) \in V(U)$ choose $y_1,\dots,y_k \notin \{x_1,\dots,x_k\} $.
Then it is easy to check that the distance between $(x_1,\dots,x_k)$ and $(y_1,\dots,y_k)$ is $k$.

Let $G$ be a directed graph such that $U=U(G)$.
Assume for the sake of contradiction that $v$ is a vertex of $G$ that is not playing its best response.
Let $A$ be the set of neighbors of $v$ (vertices with an incoming arc from $v$ or an outgoing arc to $v$)
if it had changed its strategy and played its best response.
As $v$ has degree at most $2t$, we have $|A| \leq 2t$.
Since $(2t)^k - 1 < t^k (2t-1)$, by Lemma~\ref{lem:ADelta}, there exists a vertex $u$, different from $v$, with $\dist(u,A) \geq k-1$.

Now, suppose that $v$ changes its strategy so that its neighborhood becomes $A$.
Then it is not hard to see that for any vertex $w\neq v$, the new distance between $w$ and $A$ is not less
than their old distance.
In particular, the new distance between $u$ and $A$ is at least $k-1$.
Hence the new distance between $u$ and $v$ is at least $k$,
i.e.~the local diameter of $v$ has not decreased, contradiction.
Therefore, all vertices were playing their best responses in $G$, and $G$ is an equilibrium graph.
\end{proof}

\begin{theorem}
For infinitely many $n$, there exist positive integers $b_1,b_2,\dots,b_n$,
such that there exists an equilibrium graph for $(b_1,b_2,\dots,b_n)$-BG in the MAX version with diameter $\sqrt{\log n}$.
\end{theorem}

\begin{proof}
Let $k>3$ and $t = 2^k$.
It is easy to check that we have $(2t)^k - 1 < t^k (2t-1)$.
Let $U$ be the graph given by Lemma~\ref{lem:wordgraph}, which has $n=(2^k)^k = 2 ^{k^2}$ vertices,
minimum degree at least 2, and diameter $k = \sqrt{\log n}$.
Now, let $G$ be a directed graph with $U(G)=U$ and such that the outdegree of all vertices of $G$ is at least 1.
Such a $G$ exists as the minimum degree of $U$ is larger than 1.
Then $G$ is an equilibrium graph by Lemma~\ref{lem:wordgraph} and the proof is complete.
\end{proof}


\section{An upper bound for the price of anarchy in the SUM version}
\label{sec:sumupper}
In this section we prove a general bound of $2^{O(\sqrt{\log n})}$
for the price of anarchy of bounded budget network creation games in the SUM version.
The proof, which is long and consists of several steps,
follows the line of the proof of Theorem~9 of~\cite{adhl}, but the first step is more involved.

In the following we focus on the SUM version.
For a vertex $u$ and a nonnegative integer $r$, define
$$B_r(u) = \{v : \dist(u,v) \leq r\}.$$
The first step is to prove the following theorem.
\begin{theorem}
\label{thm:first_step}
Let $u$ be a vertex of an equilibrium graph $G$, and let $r$ be a positive integer.
Assume that the subgraph induced by $B_r(u)$ is a tree $T$.
Then $r=O(\log n)$.
\end{theorem}

Assume that $u$ is chosen as the root of $T$.
Note that if every vertex in $T$ has at least two children, then
$$r = O(\log |V(T)|) = O(\log n),$$
and the theorem is proved.
Hence the problematic vertices are those with zero or one child.
Roughly speaking, in the next three lemmas, we will prove that
those vertices cannot increase the height of the tree significantly.

To prove Theorem~\ref{thm:first_step}, we need to consider weighted graphs.
We denote a weighted directed graph by $G=(V,A,w)$,
where $V$ and $A$ are the vertex set and the arc set of $G$, respectively,
and $w:V\rightarrow \mathbb{Z}^+$ assigns a weight to each vertex.
For every vertex $u$, the cost of $u$ is defined as
$$c(u) = \sum_{v\in V} w(v) \dist(u,v).$$
Note that if all vertices have unit weights, then this reduces to the original (unweighted) model.
For a subgraph $H$ of $G$ define
$$w(H) = \sum_{u\in V(H)} w(u).$$

We say that $G$ is a \emph{weak equilibrium graph} if
no vertex can decrease its cost by swapping exactly one of its edges;
more precisely, for every arc $\arc{uv} \in A$ and $x\in V$ with $\arc{ux} \notin A$,
the cost of $u$ does not decrease if the arc $\arc{uv}$ is replaced with the arc $\arc{ux}$.
Clearly every equilibrium graph is also a weak equilibrium graph.

A vertex of $G$ with degree 1 is called a \emph{leaf}.
It turns out that one should distinguish between two types of leaves:
a \emph{poor} leaf is a leaf with outdegree zero, and a \emph{rich} leaf is a leaf with outdegree one.
The poor leaves cause the most trouble and they are the reason for introducing the weights.
Let $l$ be a poor leaf in $G$, and let $\arc{ul} \in A$.
Define $G_0=(V_0,A_0,w_0)$ to be a weighted directed graph with
$$V_0 = V \setminus \{l\},\quad A_0 = A \setminus \{\arc{ul}\}, \quad
w_0(v) = \left\{ \begin{array}{c c} w(v) & \mathrm{if\ } v\neq u \\ w(u) + w(l) & \mathrm{if}\ v=u. \end{array} \right.$$
Then it can be verified that if $G$ is a weak equilibrium graph then so is $G_0$.
We say that $G_0$ is obtained by \emph{folding} the poor leaf $l$ into $u$.
The following lemma is used for handling the poor leaves.

\begin{lemma}
\label{l:poorleavefolding}
Let $G$ be a weighted weak equilibrium graph and $T$ be an induced rooted subtree of $G$ with root $z$.
Assume that
\begin{itemize}
\item every arc of $T$ is oriented away from $z$, and
\item no non-root vertex of $T$ is adjacent to a vertex outside $T$.
\end{itemize}
Then the height of $T$ is at most $1 + \log w(T)$.
\end{lemma}
\begin{proof}
For every vertex $v$ of $T$, let $T_{v}$ be the subtree of $T$ rooted at $v$.
We will prove that for every $v \in V(T)$ that is not the root, if $T_{v}$ has height $k$, then $w(T_{v}) \geq 2^{k}$.
This shows that if the height of $T$ is $h$, then $w(T) \geq 2^{h-1}$, or equivalently, $h \leq \log w(T) + 1$.

The proof is by induction on $k$.
Correctness of the case $k = 0$ follows from the fact that all weights are positive integers.
Assume that the induction hypothesis is true for $k$, and let $v \in V(T)$ be such that $T_{v}$ has height $k+1$.
Let $p$ be the parent and $x_{1},x_{2},\dots,x_{m}$ be the children of $v$.
At least one of $T_{x_{1}},T_{x_{2}},\dots,T_{x_{m}}$ has height $k$.
We may assume that the height of $T_{x_{1}}$ is $k$.
By the induction hypothesis, $w(T_{x_{1}}) \geq 2^{k}$.
We have
$$\sum_{i=2}^{m}{w(T_{x_{i}})} + w(v) \geq w(T_{x_{1}}),$$
otherwise
the vertex $p$ could decrease its cost by replacing the arc $\arc{pv}$ with the arc $\arc{px_1}$.
Thus we find
$$w(T_{v}) = w(T_{x_{1}}) +  w(T_{x_{2}}) + \dots +  w(T_{x_{m}}) + w(v) \geq 2w(T_{x_{1}}) \geq 2^{k+1},$$
which completes the proof.
\end{proof}

Note that if the conditions of the above lemma hold, then one can fold the whole subtree $T$ into the vertex $z$.
Moreover, folding this subtree does not decrease the diameter of $G$ significantly.
More precisely, the following is true.

\begin{corollary}
\label{cor:poorleavefolding}
If $G$ is a weak equilibrium graph and we perform a sequence of subtree folds on it until we obtain a new graph $G'$ with no poor leaves,
then $G'$ is also a weak equilibrium graph and
$$\diam(G') = \diam(G) - O(\log w(G)).$$
\end{corollary}

Handling rich leaves is easy, as shown by the following lemma.

\begin{lemma} \label{l:richleavesdis}
Let $G$ be a weighted weak equilibrium graph. Then the distance between any two rich leaves of $G$ is at most 2.
\end{lemma}
\begin{proof}
Let $u,v$ be two rich leaves.
By symmetry, we may assume that
$$c(u)-\dist(u,v) w(v) \leq c(v) - \dist(u,v)w(u).$$
Since $u$ is a rich leaf, it owns an arc $\arc{up}$.
If $v$ changes its strategy, by replacing its outgoing arc with the arc $\arc{vp}$, then its cost becomes
$$c(u) - \dist(u,v)w(v) + 2w(u) \leq  c(v) - \dist(u,v)w(u) + 2w(u) = c(v) + (2-\dist(u,v))w(u).$$
Since $v$ is already playing its best response, $\dist(u,v) \leq 2$, and the proof is complete.
\end{proof}

To handle the vertices of degree 2, which have one child, the following lemma will be used.

\begin{lemma} \label{l:degree2fold}
Let $G$ be a weighted weak equilibrium graph and $P$ be a path in $U(G)$ such that for every two vertices $u$ and $v$ in $P$,
the $(u,v)$-path along $P$ is the unique shortest $(u,v)$-path
(which implies, in particular, that $P$ is an induced subgraph of $U(G)$).
Then the number of edges $uv \in E(P)$ such that both $u$ and $v$ have degree 2 is $O(\log w(P))$.
\end{lemma}
\begin{proof}
The proof is similar to the proof of Theorem~\ref{t:treepoa}.
Let $P = v_0v_1\dots v_d$.
Suppose that the set of edges of $P$ whose endpoints have degree 2 is
$$S = \{v_{i_1}v_{i_1+1},v_{i_2}v_{i_2+1},\dots,v_{i_{m}}v_{i_{m}+1}\}.$$
At least half of the arcs in $S$ are in the same direction along $P$.
By symmetry, we may assume that these are the arcs
$\arc{v_{i_1}v_{i_1+1}},\arc{v_{i_2}v_{i_2+1}},\dots,\arc{v_{i_{t}}v_{i_{t}+1}}$,
where $t \geq m/2$.
For $1 \leq j < t$, if $v_{i_{j}}$ replaces its arc $\arc{v_{i_{j}} v_{i_{j}+1}}$ with the arc $\arc{v_{i_{j}} v_{i_{j}+2}}$,
then its distances to vertices in $\{v_{i_{j}+2},\dots,v_{d}\}$ decreases by one,
and its distance to $v_{i_{j}+1}$ increases by one.
Since $v_{i_{j}}$ is playing its best response, we have
\begin{displaymath}
w(v_{i_{j+1}}) \geq \sum_{k=i_{j}+2}^{d}{w(v_{k})} \geq \sum_{l=j+1}^{t}{w(v_{i_{l}+1})}
\end{displaymath}
for all $1 \leq j < t$.
Now, by the same reason as the one in the last part of the proof of the Theorem~\ref{t:treepoa}, we have $m \leq 2t = O(\log w(P))$.
\end{proof}

Now we  prove Theorem~\ref{thm:first_step}.

\begin{proof} [Proof of Theorem~\ref{thm:first_step}.]
First, make $G$ weighted by setting $w(x)=1$ for all vertices $x$. Thus $w(G) = n$.
Second, perform a sequence of subtree folds on $G$ until no poor leaves remain
(recall that a leaf is a vertex that has degree 1 in $G$).
By Corollary~\ref{cor:poorleavefolding}, the height of $T$ changes
by $O(\log w(G))$.

Third, for each edge $xy \in E(T)$ such that both $x$ and $y$ have degree 2 in $T$,
contract the edge, and repeat until no such edge exists.
We claim that doing all these contractions changes the height of $T$ by $O(\log w(G))$.
Indeed, let $v$ be any vertex in $T$ and $P$ be the unique $(u,v)$-path in $T$.
Then $P$ satisfies the conditions of Lemma~\ref{l:degree2fold},
so by this lemma, doing all these contractions
changes the distance between $u$ and $v$ by at most $O(\log w(P)) = O(\log w(G))$.

Note that since the graph was a weak equilibrium right before doing the contractions,
by Lemma~\ref{l:richleavesdis} the distance between any two rich leaves was at most 2.
The contractions do not increase the distances and do not create new leaves, so in the
final graph, there is at most one vertex in $G$ that is adjacent to leaves.
Hence the height of the obtained tree $T'$ is $O(\log |V(T')|) = O(\log w(G))$.
Therefore, the height of the original tree $T$ is $O(\log w(G))=O(\log n)$ as well.
\end{proof}

For the rest of the section, all graphs are unweighted.
The rest of the proof is similar to the proof of Theorem~9 of~\cite{adhl}.

\begin{lemma} \label{lem:frienddecrease}
Let $u,v,x$ be vertices of a graph $G$ such that the arc $\arc{uv}$ is not in $G$.
Assume that adding the arc $\arc{uv}$ to $G$ decreases the cost of $u$ by $s$, where $s > n\dist(x,u)$.
Then adding the arc $\arc{xv}$ to $G$ decreases the cost of $x$ by at least $s-n\dist(x,u)$.
\end{lemma}

\begin{proof}
For every vertex $w$, let $\improve_u(w)$ be the amount $u$ gets closer to $w$ by adding the arc $\arc{uv}$.
Similarly, let $\improve_x(w)$ be the amount $x$ gets closer to $w$ by adding the arc $\arc{xv}$.
Let $\dist^{new}(x,w)$ be the distance between $x$ and $w$ in $G\cup\arc{xv}$.
Let $W$ be the set of vertices $w$ with $\improve_u(w)>0$.
For all $w\in W$ we have
\begin{align*}
\improve_x(w) = \dist (x,w) - \dist^{new}(x,w) \geq & (\dist(u,w) - \dist(u,x)) - (1+\dist(v,w)) \\
= & \dist(u,w) - (1+\dist(v,w)) - \dist(u,x) \\
= & \improve_u(w)- \dist(u,x).
\end{align*}
Thus
\begin{align*}
\sum_{w\in V} \improve_x(w) \geq \sum_{w\in W} \improve_x(w) \geq  &
\sum_{w\in W} \left[\improve_u(w)- \dist(u,x)\right] \\
= & s - |W|\dist(u,x) \geq s - n\dist(u,x),
\end{align*}
and the proof is complete.
\end{proof}

\begin{lemma} \label{l:nlogn}
Let $G$ be a connected equilibrium graph in the SUM version that is not a tree.
Given any vertex $u$, there is an arc $\arc{xy}$ with $\dist(x,u) = O(\log n)$
and whose removal increases the cost of $x$ by at most $O(n\log n)$.
\end{lemma}
\begin{proof}
Let $r$ be the smallest positive integer such that the subgraph induced by $B_{r+1}(u)$ has a cycle.
Note that $r$ is well defined since $G$ is connected and is not a tree.
Theorem~\ref{thm:first_step} gives $r=O(\log n)$.
Consider a breadth-first search from $u$ in $U(G)$, and let $T$ denote the top $r+1$ levels of the BFS tree,
from level 0 (just $u$) to level $r+1$.
Since $B_{r+1}(u)$ has a cycle, there is an edge ${xy}\notin E(T)$ with $x,y\in V(T)$.
Assume by symmetry that the arc direction is from $x$ to $y$.
Clearly $\dist(x,u) \leq r+1 = O(\log n)$.
If the arc $\arc{xy}$ is deleted, then the distance between $x$ and any vertex
increases by at most $1+2r$,
since the shortest path can use the alternate path in $T$ from $x$ to the lowest common ancestor of $x$ and $y$ in $T$,
and then to $y$, instead of using $\arc{xy}$.
Therefore, removing $\arc{xy}$ increases the cost of $x$ by at most $n(1+2r) = O(n\log n)$.
\end{proof}

The previous lemma implies that there exist constants $p,q>0$ such that if $G$ is a connected non-tree equilibrium graph,
then for any $u \in V$,
there is an arc $\arc{xy}$ with $\dist(x,u) \leq p \log n$ and whose removal increases the cost of $x$ by at most $q n\log n$.
Using it together with Lemma~\ref{lem:frienddecrease}, we obtain the following corollary.

\begin{corollary} \label{l:orderrr}
In a connected non-tree  equilibrium graph $G$  in the SUM version,
the addition of any arc $\arc{uv}$ decreases the cost of $u$ by at most $(p+q+1) n \log n$.
\end{corollary}
\begin{proof}
Suppose for the sake of contradiction that the addition of $\arc{uv}$ decreases the cost of $u$ by more than $ (p+q+1)n \log n $.
By Lemma~\ref{l:nlogn}, there is an arc $\arc{xy}$ with $\dist(x,u) \leq p  \log n$,
and whose removal increases the cost of $x$ by at most $q n \log n$.
We show that if $x$ replaces the arc $\arc{xy}$ with the arc $\arc{xv}$,
then its cost decreases, which contradicts the fact that $G$ is an equilibrium graph.
By Lemma~\ref{lem:frienddecrease}, inserting the arc $\arc{xv}$ decreases the cost of $x$ by at least $(p+q+1) n \log n - p n \log n$,
Now, deleting the arc $\arc{xy}$ from the graph $G\cup \arc{xv}$ increases the cost of $x$ by at most $q n \log n$.
This completes the proof since
$$- (p+q+1) n \log n + p n \log n  + q n \log n < 0. \qedhere$$
\end{proof}

Now we are ready to prove the main theorem of this section,
which implies that the price of anarchy in the SUM version is $2^{O(\sqrt{\log n})}$.

\begin{theorem} \label{t:sumupperbound}
Let $G$ be a connected equilibrium graph in the SUM version.
Then the diameter of $G$ is $2^{O(\sqrt{\log n})}$.
\end{theorem}

\begin{proof}
If $G$ is a tree, then by Theorem~\ref{t:treepoa} its diameter is
$$O(\log n) = 2^{O(\log \log n)} = 2^{O(\sqrt{\log n})},$$
so we may assume that $G$ is not a tree.

Define $f(k) = \min_u |B_k(u)|$.
First, we show that
\begin{equation}
\label{eq:bk}
f(4k) \geq \min \left \{ \frac{n+1}{2}, \frac{k f(k)}{4(p+q+1) \log n} \right \}.
\end{equation}
Fix a vertex $u$, and assume that $f(4k) \leq n/2$.
Then certainly $f(3k) \leq n/2$.
Let $T$ be a maximal set of vertices at distance exactly $3k$ from $u$
subject to the distance between any pair of vertices in $T$ being at least $2k+1$.
We claim that, for every vertex $v$ of distance more than $3k$ from $u$,
the distance of $v$ from the set $T$ is at most $\dist(u,v)-k$.
Indeed, $v$ has distance $\dist(u,v)-3k$ to some vertex at distance exactly $3k$ from $u$,
and any such vertex is within distance $2k$ of some vertex of $T$, by the maximality of $T$.

Because we assumed that at least $n/2$ vertices have distance more than $3k$ from $u$,
by the pigeonhole principle, there are at least $n/(2|T|)$ such vertices $v$ whose distance from the same $t\in T$ is at most $\dist(u,v)-k$.
Adding the arc $\arc{ut}$ decreases the distances between $u$ and such vertices $v$ by $k-1$,
so improves the cost of $u$ by at least
$$(k-1) \frac {n}{2|T|} \geq \frac {kn}{4|T|}.$$
By Corollary~\ref{l:orderrr}, this improvement is at most $(p+q+1)n \log n$, so we find that
$$|T| \geq \frac{k}{4(p+q+1) \log n}.$$

Now, the sets $\{B_k(t) : t \in T\}$ are all pairwise disjoint,
all lie within distance $4k$ of $u$,
and each of them has at least $f(k)$ vertices (by the definition of $f$).
Thus
$$f(4k) \geq f(k) \frac{k}{4(p+q+1) \log n}$$
and (\ref{eq:bk}) holds.

Now we prove the theorem.
First,
$$f(2^{\sqrt{\log n}}) \geq 2^{\sqrt{\log n}}$$
simply because $G$ is connected.
Let $k$ be the smallest nonnegative integer for which
$f ( 2^{\sqrt{\log n}} 4^{k} ) > n/2$.
By (\ref{eq:bk}), for every $1\leq i<k$ we have
$$\frac{f({2^{\sqrt{\log n}} 4^i}) }{f({2^{\sqrt{\log n}} 4^{i-1}})} \geq \frac{2^{\sqrt{\log n}} 4^{i-1}}{4(p+q+1)\log n} = 2^{\Omega(\sqrt{\log n})}.$$
One can prove by induction on $i$ that for all $1\leq i<k$,
$${f({2^{\sqrt{\log n}} 4^i}) } \geq 2^{\Omega(i\sqrt{\log n})}.$$
But, since
$${f({2^{\sqrt{\log n}} 4^{k-1}} )} \leq n/2 = 2^{\log n - 1},$$
we have $k = O(\sqrt {\log n})$.
Recall that we have
$f(2^{\sqrt{\log n}} 4^k) > n/2$.
Thus for any two vertices $u$ and $v$,
$$B_{2^{\sqrt{\log n}} 4^k} (u) \cup B_{2^{\sqrt{\log n}} 4^k} (v) \neq \emptyset,$$
which means that there is a vertex $x$ within distance $2^{\sqrt{\log n}} 4^k$ of both $u$ and $v$.
So the distance between $u$ and $v$ is at most
$$2\times 2^{\sqrt{\log n}} 4^k = 2^{O(\sqrt {\log n})},$$
and the proof is complete since $u$ and $v$ were chosen arbitrarily.
\end{proof}

\begin{corollary}
The price of anarchy of any bounded budget network creation game in the SUM version is $2^{O(\sqrt{\log n})}$.
\end{corollary}
\begin{proof}
Consider a bounded budget network creation game.
If the sum of players' budgets is less than $n-1$, then clearly any realization of the game has diameter $n^2$,
so the price of anarchy is 1.
Otherwise, by Lemma~\ref{lem:connected} every equilibrium graph is connected and so by Theorem~\ref{t:sumupperbound} has diameter $2^{O(\sqrt{\log n})}$,
and this completes the proof.
\end{proof}


\section{Vertex connectivity of equilibrium graphs in the SUM version}
\label{sec:connectivity}
One of the most important issues in designing stable networks is the connectivity of the resulting network.
In this section, we find a direct connection between the budget limits and the connectivity of an equilibrium graph
in the SUM version, which shows that we can guarantee strong connectivity or a small diameter
for our network when all players have large enough budgets.
Assume that all players have budgets at least $k$ for some positive integer $k$,
and $G$ be any equilibrium graph in the SUM version.
We show that if $G$ has diameter larger than 3, then it is $k$-connected,
i.e.~removing any $k-1$ vertices does not disconnect the graph.
By Menger's theorem (see, e.g., Theorem 3.3.6 in~\cite{diestel}),
this also means that if the diameter of $G$ is larger than 3, then for any two vertices $u$ and $v$,
there exist $k$ internally disjoint paths connecting $u$ and $v$.
For a subset $A$ of vertices, let $G-A$ denote the subgraph induced by $V(G) \setminus A$.

First, we prove a lemma that will be used in the proof of the main result of this section.

\begin{lemma}
\label{lem:connectivity}
Let $G$ be an equilibrium graph in the SUM version.
Let $C$ be a subset of vertices of $G$
and $A$ be the vertex set of a connected component of $U(G-C)$.
Assume that for all $v\in A$, $\dist(v,C)=1$ and the budget of $v$ is larger than $|C|$.
Then every vertex in $A$ has local diameter at most 2.
\end{lemma}

\begin{proof}
Fix a vertex $a\in A$ and let $C'$ be the set of vertices in $C$ that are not adjacent to $a$.
Note that $C'$ might be empty.
Assume for the sake of contradiction that there exists a vertex $x$ with $\dist(a,x)>2$.
Since the budget of $a$ is larger than $|C|$ and all neighbors of $a$ are in $A\cup C$,
vertex $a$ owns at least $|C'|+1$ arcs to vertices in $A$.
Let $A'$ be a set of $|C'|+1$ vertices in $A$ to which $a$ has an arc.
If vertex $a$ deletes all these arcs
and adds arcs to the vertices $C'\cup\{x\}$ instead, then the following happens to the distances between $a$ and other vertices.
\begin{itemize}
\item The distance between $a$ and every vertex in $A'$ changes from 1 to at most 2 (recall that every vertex in $A'$ has a neighbor in $C$);
\item If a vertex in $A\setminus A'$ had distance 1 to $a$, its distance to $a$ does not change.
\item If a vertex in $A\setminus A'$ had distance more than 1 to $a$, its distance to $a$ becomes at most 2 and hence does not increase.
\item The distances between $a$ and vertices outside $A$ do not increase,
moreover its distances to vertices in $C'$ decrease by at least 1,
and its distance to $x$ decreases by at least 2.
\end{itemize}
Since $|A'| = |C'|+1$, the cost of vertex $a$ decreases by this strategy change, which contradicts the assumption that $a$
is playing its best response.
Therefore every $a\in A$ has local diameter at most 2 in $G$.
\end{proof}

Now we  prove the main result of this section.

\begin{theorem}
\label{thm:k-connectivity}
Suppose that $G$ is an equilibrium graph in the SUM version and all vertices have budgets at least $k$.
If $G$ has diameter greater than 3, then it is $k$-connected.
\end{theorem}
\begin{proof}
Let $C$ be a minimum cardinality subset of vertices such that $U(G-C)$ is disconnected and let $k'$ be the size of $C$.
If $k' \geq k$ then $G$ is $k$-connected and the theorem is proved, so we may assume that $k' < k$,
and we need to show that the diameter of $G$ is at most 3.
Let $A$ be the vertex set of the connected component of $U(G-C)$ with minimum size,
and let $B = V(G) \backslash (A\cup C)$.

We claim that every vertex in $A$ is adjacent to some vertex in $C$.
Let $u\in A$ be a vertex with maximum distance from $C$, and let $l=\dist(u, C)$.
If $l=1$, then the claim is proved, so assume that $l>1$.
Since the budget of $u$ is larger than $k'$, there is a set $U\subseteq A$
of size $k'$ to each vertex of which $u$ has an arc.
Let us study what happens if $u$ changes its strategy
by removing the arcs to $U$ and adding arcs to all vertices in $C$.
Note that after this change, the distance between $u$ and any vertex in $A$ becomes at most $l+1$.
\begin{itemize}
\item The distance between $u$ and every vertex in $C\cup B$ decreases by at least $l-1$,
giving a total decrease of at least $(|B| + k')(l-1)$.
\item The distance between $u$ and every vertex in $U$ changes from 1 to at most $l+1$,
giving a total increase of at most $k'l$.
\item The distance between $u$ and every vertex $v\in A \setminus (U\cup\{u\})$ increases by at most $l-1$:
if $v$ was adjacent to $u$, it remains adjacent,
and if $v$ had distance at least 2, its distance becomes at most $l+1$.
Therefore, the total increase is at most $(|A|-k'-1)(l-1)$.
\end{itemize}
Therefore, the total change in the cost of $u$ is less than or equal to
$$-(|B| + k')(l-1) + k'l + (|A|-k'-1)(l-1) =  (|A|-|B|)(l-1) + k'(2-l) + 1 - l  < 0,$$
which contradicts the fact that $u$ is playing its best response.
Thus $l=1$ and every vertex in $A$ is adjacent to some vertex in $C$.

Now by Lemma~\ref{lem:connectivity}, every vertex in $A$ has local diameter 2.
Let $D$ be any connected component of $U(G-C)$ other than $A$.
Every vertex in $D$ is at distance 2 from $A$ and so is adjacent to some vertex in $C$.
Therefore, by Lemma~\ref{lem:connectivity}, every vertex in $D$ has local diameter 2.
Since $D$ was arbitrary, we conclude that every vertex in $U(G-C)$ has local diameter 2.
To complete the proof we just need to show that every vertex in $C$ has local diameter at most 3.
Let $x\in C$ be arbitrary. Since the budget of $x$ is larger than $|C|$, it has a neighbor $u$ outside $C$.
Since the local diameter of $u$ is 2, the local diameter of $x$ is at most 3 and the proof is complete.
\end{proof}


\section{Concluding remarks}
\label{sec:conclusion}
We considered network creation games in which every vertex
has a specific budget for the number of vertices it can establish links to,
and the distances are computed based on the underlying graph of the resulting network.
Two versions of this game were defined, depending on whether each vertex wants to minimize
its maximum distance or sum of distances to other vertices.
We proved that Nash equilibria exist in both versions by explicitly constructing them.
A natural question is, if the game starts from an arbitrary position and
the players keep on improving their strategies, does the game converge to an equilibrium?
If yes, then how quickly does it converge?
Note that Laoutaris~et~al.~\cite{laoutaris} demonstrated that
converging to a pure Nash equilibrium is not guaranteed in their game,
by constructing an explicit loop.

We studied the diameter of equilibrium graphs
and the price of anarchy of this game and
found asymptotically tight bounds in two extreme cases,
namely when the sum of players' budgets is $n-1$ (the threshold for connectivity),
and when all players have unit budgets.
There are other special cases that might be interesting,
for example the cases in which all players have the same budget $B > 1$.

We observed a non-monotone property of this game (in the MAX version):
when all budgets are 1, the diameter of equilibrium graphs (and the price of anarchy) is $O(1)$;
but when all budgets are at least 1, there are instances of the game
in which the diameter of equilibrium graphs is $\Omega(\sqrt{\log n})$,
and so the price of anarchy is also $\Omega(\sqrt{\log n})$.
We proved that this happens in the MAX version,
and it is interesting to know if it can also happen in the SUM version.
More precisely, when all players have positive budgets,
is the diameter of equilibrium graphs in the SUM version bounded?

In the SUM version, we proved a general upper bound of $2^{O(\sqrt{\log n})}$
for the diameter of equilibrium graphs (and the price of anarchy).
The proof is complicated and based on showing a certain expansion property for equilibrium graphs.
The bound $2^{O(\sqrt{\log n})}$ seems strange and probably is not tight.
Finally, we showed an interesting connection between the minimum budget
available to the players and the connectivity of the resulting network.
More precisely, if each vertex has budget at least $k$,
then any equilibrium graph in the SUM version has diameter less than 4, or is $k$-connected.


\vspace{0.5 cm}

\noindent {\bf Acknowledgement.}
The authors are thankful to Nastaran Nikparto for her useful comments throughout the preparation of this paper,
and to Laura Poplawski Ma for suggesting the idea of defining the distance of vertices in different components to be $C_{\mathrm{inf}} = n^2$,
which resulted in slightly cleaner arguments.

\bibliographystyle{acm}
\bibliography{networkcreationgames}

\end{document}